\pdfoutput=1
\documentclass[11pt]{article}
\usepackage{amsmath}
\usepackage{amsthm}

\usepackage{tikz}
\usepackage{fullpage}
\usepackage{xcolor}
\usepackage{cite}
\usepackage{color}
\usepackage{graphicx}
\usepackage{calc}
\usepackage{xstring}
\usepackage{url}
\usepackage[pdfencoding=auto]{hyperref}
\usepackage{cleveref}
\usepackage[noend]{algorithmic}
\usepackage{algorithm}
\usepackage{subcaption}

\newtheorem{theorem}{Theorem}

\newtheorem*{fact*}{Fact}
\renewenvironment{proof}{\noindent{\bf Proof:}}{\hspace*{\fill}\rule{6pt}{6pt}\bigskip}

\hypersetup{
   colorlinks=true,
   linkcolor=blue,
   urlcolor=blue,
   linktoc=all,
   citecolor=blue
}

\definecolor{okabe1}{HTML}{000000}
\definecolor{okabe2}{HTML}{E69F00}
\definecolor{okabe3}{HTML}{56B4E9}
\definecolor{okabe4}{HTML}{009E73}
\definecolor{okabe5}{HTML}{F0E442}
\definecolor{okabe6}{HTML}{0072B2}
\definecolor{okabe7}{HTML}{D55E00}
\definecolor{okabe8}{HTML}{CC79A7}

\renewcommand{\emph}[1]{\textit{\textbf{#1}}}
\let\epsilon\varepsilon

\date{}
\title{Making Quickhull More Like Quicksort:\\
A Simple Randomized Output-Sensitive Convex Hull Algorithm}

\author{Michael T. Goodrich \\ University of California, Irvine, USA \and
Ryuto Kitagawa \\ University of California, Irvine, USA}

\begin{document}

\setcounter{page}{0}
\maketitle
\begin{abstract}
In this paper, we present \emph{Ray-shooting Quickhull}, which is a simple,
randomized, output-sensitive version of the Quickhull algorithm for
constructing the convex hull of a set of $n$ points in the plane.
We show that the randomized Ray-shooting Quickhull algorithm
runs in $O(n\log h)$ expected time, where $h$
is the number of points on the boundary of the convex hull.
Keeping with the spirit of the original Quickhull algorithm,
our algorithm is quite simple and is, in fact,
closer in spirit to 
the well-known randomized Quicksort algorithm.
Unlike the original Quickhull algorithm, however,
which can run in $\Theta(n^2)$ time for some input distributions,
the expected performance bounds for the randomized Ray-shooting Quickhull
algorithm match or improve
the performance bounds of more complicated algorithms. 
Importantly, the expectation in our output-sensitive 
performance bound does not depend on
assumptions about the distribution of input points.
Still, we show that, like the deterministic Quickhull algorithm,
our randomized Ray-shooting Quickhull algorithm
runs in $O(n)$ expected time for $n$ points chosen uniformly at random from
a bounded convex region.
We also provide experimental evidence that the randomized Ray-shooting
Quickhull algorithm is on par or faster than deterministic Quickhull
in practice, depending on the input distribution.
\end{abstract}
\thispagestyle{empty}

\clearpage
\section{Introduction}
The convex hull problem is arguably the most-studied problem in computational
geometry; e.g., see Seidel~\cite{seidel2017convex}.
In the two-dimensional version of this 
problem, one is given a set, $S$, of $n$ points in the plane and asked to
output a representation of the smallest convex polygon that contains
the points in $S$.
(See Figure~\ref{fig:convex}.)
It is easy to see that the output size, $h$, can range from $3$ (when the
convex hull is a triangle) to $n$ (when all the points of $S$ 
are on the boundary of the convex hull).
The asymptotically fastest convex hull algorithms
are \emph{output sensitive}, meaning that their running time depends
on both $n$ and $h$, with the best such algorithms running in 
$O(n\log h)$ time; e.g., see Kirkpatrick and Seidel~\cite{kirkpatrick}
and Chan~\cite{chan}.
Unfortunately, although these output-sensitive algorithms are asymptotically
optimal, they are somewhat complicated and 
tend to be inefficient in practice; e.g., see
McQueen and Toussaint~\cite{mcqueen}.
Thus, it would be desirable to have a simple, practical, output-sensitive
convex hull algorithm.

\begin{figure}[hbt]
\centering
\includegraphics[width=4.2in, trim = 1.2in 2.5in 7.9in 1.6in, clip]{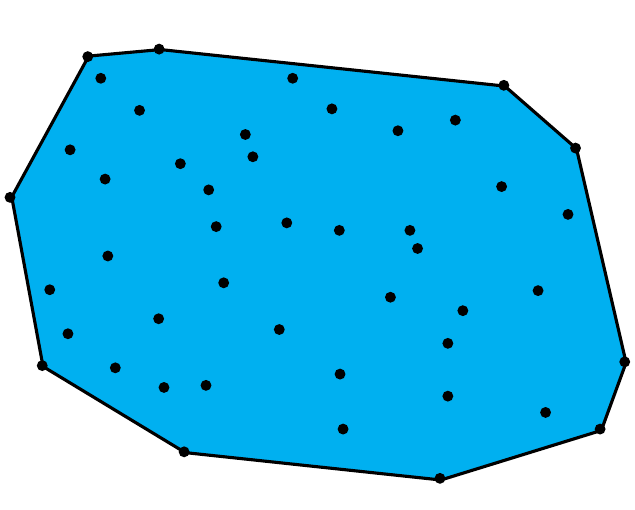}
\caption{\label{fig:convex} A two-dimensional convex hull.}
\end{figure}

\subsection{Related Prior Work}
Quickhull is a simple algorithm for finding
the convex hull of a set of $n$ points in the plane,
which has been taught many times
to undergraduates; see, e.g., Mount~\cite{mount},
Preparata and Shamos~\cite{preparata}, and O'Rourke~\cite{orourke}.
The Quickhull algorithm was first published (apparently independently)
in the late 1970s
by Eddy~\cite{eddy}, Bykat~\cite{bykat}, and Green and Silverman~\cite{green},
but it wasn't given the name
``Quickhull'' until years later; see, e.g.,
Preparata and Shamos~\cite{preparata}, 
Greenfield~\cite{greenfield} and Barber, Dobkin, and Huhdanpaa~\cite{barber}.
According to Google Scholar, the Quickhull 
paper by Barber {\it et al.}~\cite{barber}
has been cited over 7,000 times.

The Quickhull algorithm is deterministic and runs in $O(n^2)$ time in the worst
case, but performs well in practice for a variety of input
distributions; see, e.g., Gamby and Katajainen~\cite{gamby2018convex},
Mitura, {\v{S}}imecek, and Kotenkov~\cite{mitura}, and
Overmars and van~Leeuwen~\cite{overmars}.
Intuitively, unlike the output-sensitive
convex hull algorithms of Kirkpatrick and Seidel~\cite{kirkpatrick}
and Chan~\cite{chan},
Quickhull often has fast running times due to its ability to prune many input
points early in the recursive calls of its divide-and-conquer structure.
Still, its deterministic nature 
and worst-case inefficient performance is not ideal.
Indeed, in his widely-used lecture notes, Mount~\cite{mount} states
that ``unlike QuickSort, there is no obvious way to convert [Quickhull] 
into a randomized algorithm with $O(n \log n)$ expected running time.''
The goal of this paper, therefore, is to provide a simple randomized version
of the Quickhull algorithm suitable for teaching to undergraduates that
has an efficient expected running time and is more like the
randomized Quicksort algorithm 
(e.g., see~\cite{clrs,goodrich,iliopoulos,tardos}).

There has, in fact, been prior work on designing randomized 
versions of the Quickhull algorithm, 
but none of the published algorithms are simple.
Wenger~\cite{wenger} presents a randomized Quickhull algorithm, which is based
on pairing input points and pruning them based on the the slopes of the lines
determined by the pairs. Wenger's algorithm is not simple, however, and he admits
that its running time has large constant factors.
Bhattacharya and Sen~\cite{bhattacharya} improve the constant factors
for this approach, but their algorithm is still not as simple as the original
Quickhull algorithm and the constant factors in their running time analysis
are not small.
Chan, Snoeyink, and Yap~\cite{chan1995output} also provide a randomized 
convex hull algorithm based on pairing points and pruning points based on slopes,
but it too is not simple.
In addition, 
Br{\"o}nnimann, Iacono, Katajainen, Morin, Morrison, and 
Toussaint~\cite{bronnimann2002place}
provide an in-place version of this approach.
In contrast, a different output-sensitive 
convex hull algorithm by Chan~\cite{chan} is simpler,
but it is based on using several algorithmic ``tricks''
and it cannot be considered a version of Quickhull.

\subsection{Our Results}
In this paper, we present a
simple randomized \emph{Ray-shooting Quickhull} algorithm
for constructing the convex hull of a set of $n$ points in the plane.
We show that our randomized Ray-shooting Quickhull algorithm
runs in $O(n\log h)$ expected time, where $h$
is the number of points on the boundary of the convex hull.
Moreover, we show that the constant factor in this expected running time is small.

Our algorithm is closer in spirit to 
the well-known randomized Quicksort algorithm 
(e.g., see~\cite{clrs,goodrich,iliopoulos,tardos}),
in that it involves repeatedly picking a ``pivot'' point at random
and splitting subproblems based on how this pivot divides the points.
A crucial component in our algorithm is that this splitting is done
by a ray-shooting operation, where we shoot a ray from the pivot to 
determine where the ray would cross the boundary of the convex hull.
We provide a novel, simple ray-shooting algorithm to perform this step
and we show that it performs an expected number of orientation tests
that is at most $2n$.

We also provide an explicit construction of a point set
where each point can be represented in floating point using $O(\log n)$
bits that shows that the original Quickhull algorithm can require
$\Omega(nh)$ time.
Unlike the original Quickhull algorithm, 
which can run in $\Theta(n^2)$ time for some input distributions,
the expected running time for the randomized Ray-shooting Quickhull
algorithm matches or improves
the expected running time of more complicated algorithms. 
For example, we show that, like the deterministic Quickhull algorithm,
our randomized Ray-shooting Quickhull algorithm
runs in $O(n)$ expected time for $n$ points chosen uniformly at random from
a bounded convex region.
We also provide an experimental analysis as well.

\section{Quickhull}
In this section, we review the Quickhull algorithm.
Suppose that we are given a set, $S$, of $n$ points in the plane.
The Quickhull algorithm begins by finding a point, $p\in S$, with minimum
$x$-coordinate, and a point, $r\in S$, with maximum $x$-coordinate.
Clearly, $p$ and $r$ are on the convex hull of $S$.
Quickhull is a divide-and-conquer algorithm, where at each call we
are given a line segment, $\overline{pr}$, where $p$ and $r$ are on the convex
hull, and a subset, $S'$, of $S$ of points on one side of $\overline{pr}$.
Initially, there are two subproblems, one for the points above the initial
$\overline{pr}$ segment and one for the points below $\overline{pr}$.
Next, for each recursive call, we have a set of points, $S'\subseteq S$, inside
a triangle with base $\overline{pr}$, for which
Quickhull determines the point, $q$ in $S'$, that is farthest 
from the segment $\overline{pr}$. 
Note that $q$ must be on the boundary of the convex hull. 
Quickhull then prunes away any points of $S'$ inside the triangle
$(p,q,r)$, since they cannot belong to the boundary of the convex hull
of $S$.
Next, if we view $\overline{pr}$ as being horizontal, then 
we partition the remaining points of $S'$ into those that 
are above $\overline{pq}$ and $\overline{qr}$ and in bounding triangles
defined by the tangents, respectively, and we recursively
solve the problem for each of these subsets if they are nonempty.
See Figure~\ref{fig:quickhull}.

\begin{figure}[hbt]
\centering
\includegraphics[width=6in, trim = 1.15in 1.75in 1in 1.5in, clip]{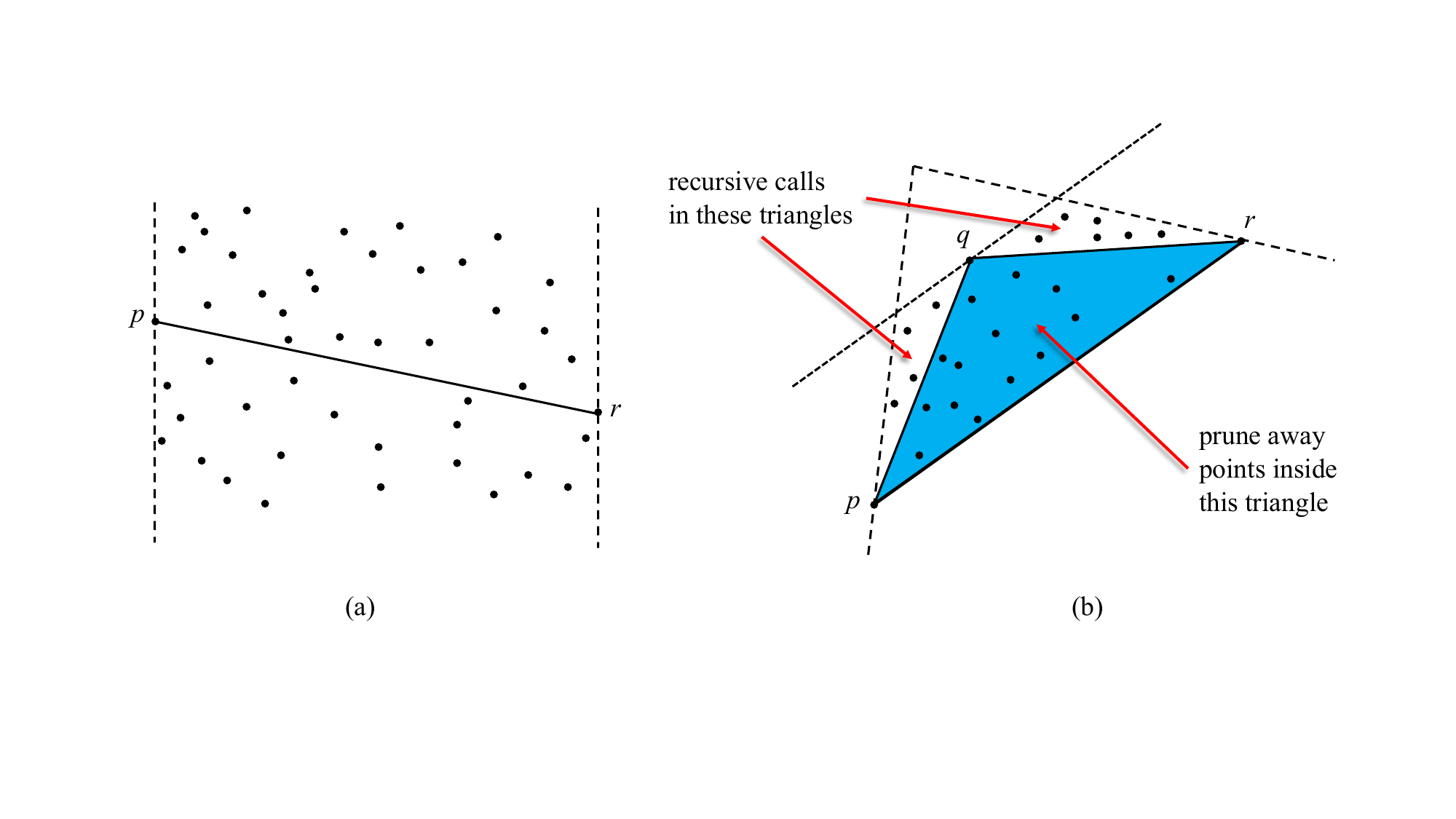}
\caption{\label{fig:quickhull} Illustrating the Quickhull algorithm.
(a) the initialization step; (b) a recursive call.}
\end{figure}

Quickhull is clearly deterministic and
it is easy to see that the worst-case running time of Quickhull 
is $O(n^2)$, much like the worst-case running time 
of Quicksort
(e.g., see~\cite{ahu,clrs,goodrich,hoare1962quicksort,iliopoulos,tardos}).
Indeed, it is also easy to show that the
worst-case running time of Quickhull is $O(nh)$, 
where $h$ is the number of points in the convex hull, since each
recursive call runs in $O(n)$ time and is guaranteed to find a distinct
point on the convex hull.
Furthermore, there are inputs that require Quickhull to have this running time,
as we show in Section~\ref{sec:experiments}.

\section{The Randomized Ray-Shooting Quickhull Algorithm}
In this section, we present the randomized Ray-shooting Quickhull 
algorithm.
Suppose we are given a set, $S$, of $n$ points in the plane.
The randomized Ray-shooting Quickhull algorithm begins as in the Quickhull
algorithm, by finding a point, $p\in S$, with minimum
$x$-coordinate, and a point, $r\in S$, with maximum $x$-coordinate.
Clearly, $p$ and $r$ are on the convex hull of $S$.
Randomized Ray-shooting Quickhull is 
a divide-and-conquer algorithm, where at each call we
are given a line segment, $\overline{pr}$, where $p$ and $r$ are on the convex
hull, and a subset, $S'$, of $S$ of points on one side of $\overline{pr}$.
Initially, there are two subproblems, one for the points above the initial
$\overline{pr}$ segment and one for the points below $\overline{pr}$.
Next, for each recursive call, we have a set of points, $S'\subseteq S$, 
inside a triangle with base $\overline{pr}$, for which
we choose a point, $q$ in $S'$, uniformly at random.
Note that $q$ is not necessarily on the boundary of the convex hull, 
but we nevertheless view $q$ as a ``pivot,'' as 
in the well-known randomized Quicksort algorithm
(e.g., see~\cite{clrs,goodrich,iliopoulos,tardos}).
We then perform a \emph{ray-shooting} query, $\vec{R}$, 
for $S'$ from the pivot, $q$, directed
in the normal (perpendicular) direction from $\overline{pr}$. 
This ray-shooting query determines 
a pair of points, $(s,t)$, such that $\overline{st}$ is an edge of the convex
hull of $S'$ that is intersected by $\vec{R}$.
This edge is often called a \emph{bridge}~\cite{kirkpatrick}.
If $q$ is itself on the convex hull and would have been the point
chosen in the original Quickhull algorithm, then 
$s=t=q$, and we view $\overline{st}$ as a zero-length line segment with
slope equal to the slope of $\overline{pr}$.
We then prune away any points of $S'$ inside the polygon,
$(p,s,t,r)$, since these points 
cannot belong to the boundary of the convex hull of $S$.
Next, if we view $\overline{pr}$ as being horizontal, then 
we partition the remaining points of $S'$ into those that 
are above $\overline{ps}$ and $\overline{tr}$, respectively, and we recursively
solve the problem for each of these subsets if they are nonempty. 
The subproblems have bounding triangles defined by
the original boundaries, the line $\overline{st}$, and the line segments
$\overline{ps}$ and $\overline{tr}$, respectively.
See Figure~\ref{fig:rquickhull}.

\begin{figure}[hbt]
\centering
\includegraphics[width=6in, trim = 1.15in 1.75in 1in 1.5in, clip]{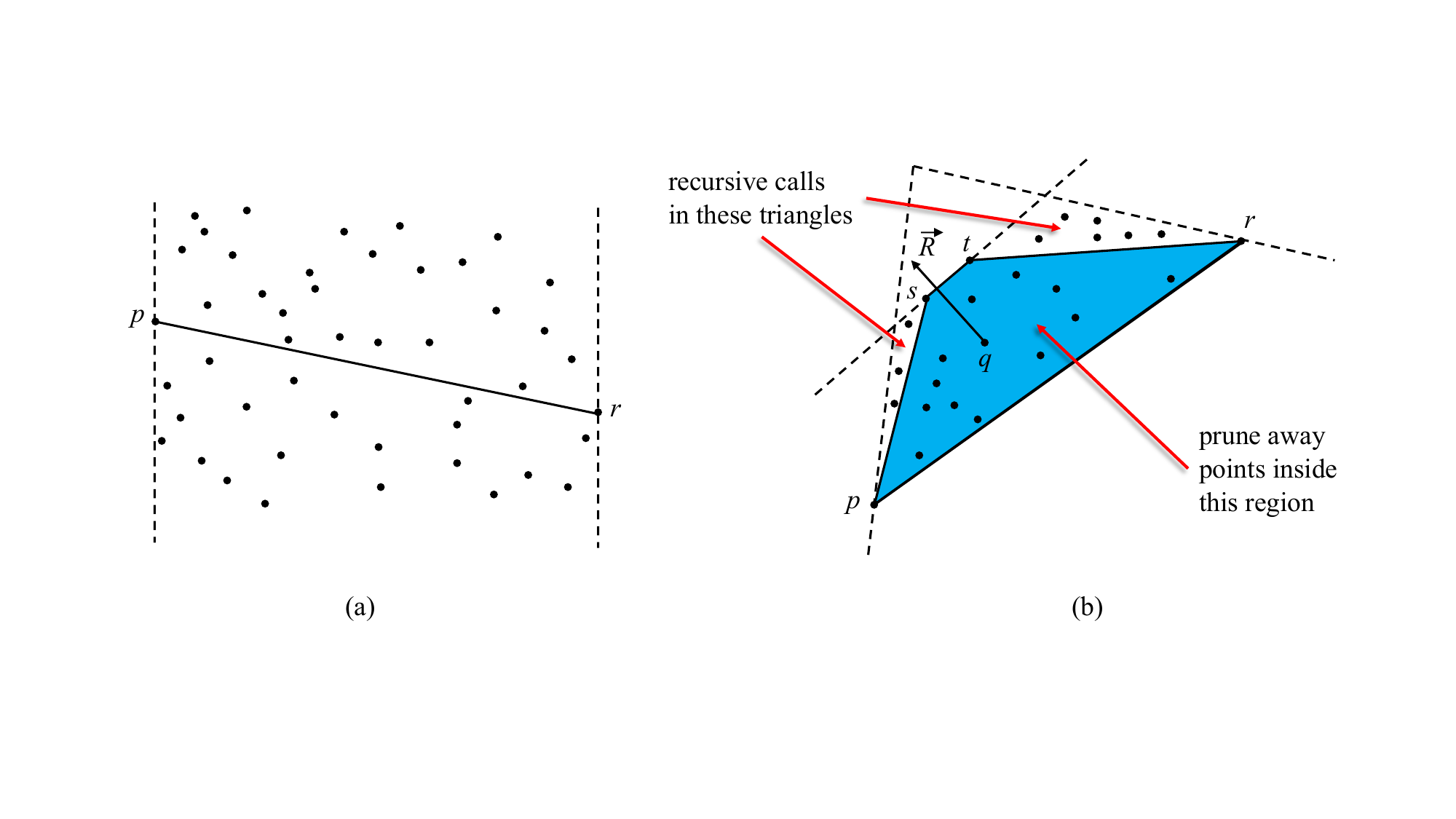}
\caption{\label{fig:rquickhull} Illustrating the 
randomized Ray-shooting Quickhull algorithm. (a) the initialization step;
(b) a recursive call.}
\end{figure}

\subsection{Ray-shooting Queries}
Let us next describe how to efficiently implement the
ray-shooting queries performed in the randomized Ray-shooting 
Quickhull algorithm.
We begin by observing that the ray-shooting query used
in our randomized Ray-shooting Quickhull algorithm can 
be answered in $O(n)$ time
using a bridge-finding algorithm by Kirkpatrick and Seidel~\cite{kirkpatrick},
but their method is fairly complex and does not have a small constant
factor in its running time.
By a point-line duality, the ray-shooting queries performed
in the randomized Ray-shooting Quickhull algorithm can also
be solved by a reduction to two-dimensional linear programming, 
but doing this problem conversion (and its reversal) adds needless 
complications to the algorithm, as well as needing to deal with issues
that arise in two-dimensional linear programming that do not arise
in answering the ray-shooting queries used in our
randomized Ray-shooting Quickhull algorithm.
Thus, for the sake of completeness and simplicity, let us
describe here a direct randomized incremental algorithm 
for answering these types of ray-shooting queries. 
Our algorithm is an adaptation and simplification 
of a randomized algorithm by Seidel~\cite{seidel} 
for two-dimensional linear programming; see also,
e.g., de~Berg {\it et al.}~\cite{deberg}.

Suppose that we are given a set, $S$, of $n$ points and point, $q$.
Without loss of generality, let us assume that base segment,
$\overline{pr}$, is horizontal and the ray, $\vec{R}$, we are shooting from
$q$ is vertical and pointing upward; hence, we are interested in
finding a bridge on the upper hull of $S\cup\{q\}$.
Let $R_l$ denote the halfplane to the left of $\vec{R}$
and $R_r$ denote the halfplane to the right of $\vec{R}$.
Our algorithm is shown in Algorithm~\ref{alg:rayshoot}.
See Figure~\ref{fig:ray}.

\begin{algorithm}[hbt]
\vspace*{6pt}
\textbf{Algorithm} RayShoot($S,q$):
\begin{algorithmic}[1]
\STATE Let $(s,t)\leftarrow (q,q)$ be our initial candidate convex hull
(degenerate) bridge edge.
\STATE Let $(s,t)$ initially define a horizontal line, $\overline{st}$, through $q$.
\STATE Let $R_l$ denote the left halfplane defined by the vertical line through $q$.
\STATE Let $R_r$ denote the right halfplane defined by the vertical line through $q$.
\STATE Let $S_l=S_r=\{q\}$ be the set of points processed so far 
that are respectively in $R_l$ and $R_r$.
\STATE Randomly permute the points in $S=\{p_1,p_2,\ldots,p_n\}$.
\label{permute}
\FOR {$i\leftarrow 1$ \textbf{to} $n$}
\IF {$p_i$ is above the line $st$}
\label{if-state}
\IF {$p_i$ is in $R_l$}
\STATE Find the point, $t'$, in $S_r$
such that $p_i t'$ minimizes the angle with the $x$-axis.
\label{search1}
\STATE Let $(s,t)=(p_i,t')$.
\ELSE
\STATE Find the point, $s'$, in $S_l$
such that $s' p_i$ minimizes the angle with the $x$-axis.
\label{search2}
\STATE Let $(s,t)=(s',p_i)$.
\ENDIF
\ENDIF
\STATE {\bf if} $p_i\in R_l$ {\bf then} add $p_i$ to $S_l$.
\STATE {\bf if} $p_i\in R_r$ {\bf then} add $p_i$ to $S_r$.
\ENDFOR
\end{algorithmic}
\caption{\label{alg:rayshoot} 
Given a set, $S$, of $n$ points in the plane and a point, $q$,
the RayShoot algorithm 
finds the edge of the upper hull of $S\cup\{q\}$
intersected by a vertical ray, $\vec{R}$, from $q$.
If $q$ is on the convex hull of $S\cup\{q\}$, then the algorithm
returns either a degenerate edge, $(q,q)$, or an edge of 
the convex hull that includes $q$.
We describe the algorithm assuming there is a horizontal
base edge and all the points of the input set, $S$, are above it.
}
\end{algorithm}

\begin{figure}[hbt]
\centering
\includegraphics[width=6.2in, trim = 0.9in 3in 1.5in 1.4in, clip]{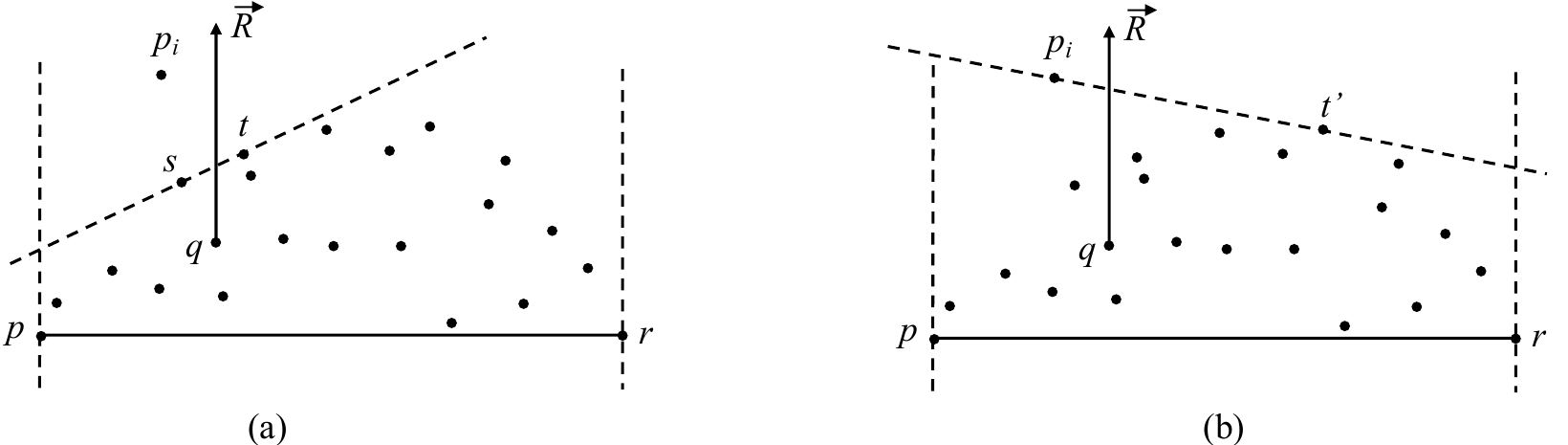}
\vspace*{-8pt}
\caption{\label{fig:ray} Illustrating the 
ray-shooting algorithm. (a) before considering the point $p_i$; (b) after
updating the edge $(s,t)$ to be $(p_i,t')$.}
\vspace*{-26pt}
\end{figure}

\clearpage
\begin{theorem}
Algorithm RayShoot is correct.
\end{theorem}
\begin{proof}
The correctness of RayShoot follows by a simple inductive argument,
where the induction hypothesis is that
the edge, $(s,t)$, is the bridge edge
of the upper hull of $\{q,p_1,p_2,\ldots,p_i\}$ that intersects $\vec{R}$.
This is true initially, for $i=0$, since 
$(s,t)=(q,q)$ and the line $\overline{st}$ is a horizontal line.
For the induction step, $i\ge 1$, suppose the induction hypothesis is true for 
iteration $i-1$.
If the point, $p_i$, is below the line $\overline{st}$, then $p_i$ does
not invalidate the invariant that $(s,t)$ is a bridge edge intersecting
$\vec{R}$.
So suppose $p_i$ is above the line $\overline{st}$. W.l.o.g., suppose
$p_i\in R_l$.
Then we replace $(s,t)$ with the edge $(p_i,t')$, such that
$t'\in S_r$ and the line $\overline{p_it'}$ is a tangent line for
the convex hull of the set
$R_r\cap\{q,p_1,p_2,\ldots,p_{i-1}\}$.
Further, by the induction hypothesis, at the beginning of 
iteration $i$, none of the points
in $R_l\cap\{q,p_1,p_2,\ldots,p_{i-1}\}$ were above the line 
$\overline{st}$. Also, 
note that
$t\in R_r\cap\{q,p_1,p_2,\ldots,p_{i-1}\}$, and $t$ was considered
in the search for $t'$ performed in line~\ref{search1}.
Thus, there can be no point
in $R_l\cap\{q,p_1,p_2,\ldots,p_{i-1}\}$ above the line 
$\overline{p_it'}$.
Therefore, we satisfy the induction invariant for the next iteration,
which establishes the proof.
\end{proof}

\begin{theorem}
\label{thm:time}
Algorithm RayShoot performs at most $2n$
orientation tests in expectation.
\end{theorem}
\begin{proof}
The running time analysis follows by a simple backwards analysis.
Let $X_i$ be a 0-1 random variable that is 1 if and only if the
condition in line~\ref{if-state} in the ray-shooting algorithm is true.
Since the searching operations in lines~\ref{search1} 
and~\ref{search2}
use an orientation test for each member of $S_r$ (resp., $S_l$),
the total number of orientation tests performed by RayShoot is 
at most
\[
\sum_{i=1}^n i X_i.
\]
By the linearity of expectation, 
\[
E\left[\sum_{i=1}^n i X_i \right]
=
\sum_{i=1}^n i P_i,
\]
where $P_i$ is the probability that the point $p_i$ is above the line $st$.
Now consider the iterations of RayShoot backwards, and note that $p_i$
will satisfy the condition in line~\ref{if-state} if it is one of the two
points that defines the edge of the convex hull of 
$q\cup \{p_1,p_2,\ldots,p_i\}$ intersecting $\vec{R}$. Thus,
$P_i\le 2/i$, which implies that the expected number of orientation
tests performed by RayShoot is 
\[
\sum_{i=1}^n {i} \cdot \frac{2}{i} = 2n.
\]
This completes the proof.
\end{proof}

In practice, we would expect the size of $S_l$ or $S_r$ in iteration $i$
to be closer to $i/2$, since the pivot $q$ is chosen uniformly at random;
hence, the upper bound of $2n$ in Theorem~\ref{thm:time} is conservative.
In any case, the expected running time 
for RayShoot is $O(n)$ with a small constant factor.

\subsection{Analyzing the Randomized Ray-shooting Quickhull Algorithm}
In this subsection, we analyze the expected running time of 
the randomized Ray-shooting Quickhull algorithm.

\begin{theorem}
Given a set, $S$, of $n$ points in the plane,
the randomized Ray-shooting Quickhull algorithm constructs the convex hull
of $S$ in $O(n\log h)$ expected time, where $h$ is the number of points of $S$
on the convex hull.
\end{theorem}
\begin{proof}
The proof is an adaptation of an analysis of the expected running time of 
the Quicksort algorithm~\cite{ahu,iliopoulos,kirkpatrick}.
Let $T(n,h)$ denote the expected running time of the randomized Ray-shooting 
Quickhull algorithm on 
an instance of size $n\ge 2$ with hull size $h\ge 2$.
Also, to simplify the notation, let $T(0,h)=0$ and $T(1,h)=0$.
Then, by the way a problem instance in the randomized Ray-shooting Quickhull 
algorithm is divided,
there is a constant $c\ge1$, such that
the general case is as follows:
\[
T(n,h) \le cn + \frac{1}{n} \sum_{i=0}^{n-1} \max_{h_1+h_2=h} 
      \{ T(i,h_1) + T(n-i-1,h_2)\} ,
\]
where, by Theorem~\ref{thm:time}, $c=2$ if we are focused on counting
orientation tests.
We claim that 
there is a constant, $d\ge1$, such that $T(n,h)\le dn$, for
$n\ge 2$ and $h=1,2$, 
and $T(n,h)\le dn\log h$ otherwise;\footnote{W.l.o.g.,
   we also assume in this proof that ``$\log$'' is the natural logarithm.}
hence, by this induction hypothesis,
\[
T(n,h) \le cn + \frac{1}{n} \left( 2d(n-1)\log(h-1)+\sum_{i=1}^{n-2} \max_{h_1+h_2=h} 
      \{ di\log h_1 + d(n-i-1)\log h_2\} \right).
\]
By elementary calculus,
the righthand side is maximized with
$h_1= ih/n$ and $h_2=(n-i-1)h/n$.
Thus, 
\begin{eqnarray*}
T(n,h) &\le& cn + \frac{d}{n} \left(2(n-1)\log(h-1)+\sum_{i=2}^{n-2} 
      \left( i\log (ih/n) + (n-i-1)\log ((n-i-1)h/n)\right)\right) \\
 &\le& cn + \frac{2d}{n} \sum_{i=1}^{n-1} i\log (ih/n) \\
 &=& cn + \frac{2d}{n} \sum_{i=1}^{n-1} i\log i  
  + \frac{2d}{n}(\log h) \sum_{i=1}^{n-1} i
  - \frac{2d}{n}(\log n) \sum_{i=1}^{n-1} i.
\end{eqnarray*}
By another application of calculus,
\[
\sum_{i=1}^{n-1} i\log i \le \int_1^n (x\log x)dx \le (n^2/2)\log n - n^2/4+1/4.
\]
Also, it is well-known that $\sum_{i=1}^{n-1} i = n(n-1)/2$.
Therefore,
\begin{eqnarray*}
T(n,h) &\le& cn + dn\log n - dn/2 + d/(2n) + d(n-1)\log h - d(n-1)\log n \\
 &\le& dn\log h,
\end{eqnarray*}
for $d = 2c+1$.
\end{proof}

Thus, the constant in the expected
running time for our randomized Ray-shooting Quickhull algorithm is small.

\section{Analyses for Various Input Distributions}
\label{sec:experiments}
In this section, we provide analyses of the deterministic Quickhull and 
randomized Ray-shooting Quickhull algorithms for various input distributions.

\subsection{A Lower Bound for Deterministic Quickhull}
We begin with a lower bound distribution for the deterministic Quickhull algorithm,
which makes explicit and generalizes implicit constructions
of Fournier~\cite{fournier}
and D{\'evai} and Szendr{\'e}nyi~\cite{devai}.

\begin{theorem}
\label{thm:lower}
For any $n\ge 3$ and $3\le h\le n$, there is a set of $n$ points
with a convex hull of size $h$ 
that causes Quickhull to run in $\Omega(nh)$ time.
\end{theorem}
\begin{proof}
Consider the set, 
\[
S= (0,0)\cup \{(2^i,2^{2i}\}, \mbox{~~for $i=1,2,\ldots,h-1$}.
\]
$S$ can be viewed as a set of $h$ exponentially separated points
on the $x$-axis that are then projected onto the parabola, 
\[
y=x^2.
\]
Since the points of $S$ are in convex position, they are all on the boundary
of the convex hull of $S$.
Next, let $S'=S\cup T$,
where $T$ is a set of $n-h-1$ points in the interior of the triangle,
$((0,0),(1,1),(2,4))$.
We claim that the $j$-th call to Quickhull on $S'$ will have base edge,
$((0,0),(2^i,2^{2i}))$, where $i=h-j$.
This is clearly true initially. Assume this is true inductively for 
call $j$, and consider call $j+1$.
The edge, $((0,0),(2^i,2^{2i}))$, for $i=h-j$ has slope $2^i$;
hence, this instance of Quickhull will choose the point that has a 
tangent with this slope.
Since the derivative of $f(x)=x^2$ is $2x$, by elementary calculus,
this point of tangency 
is the point, $(x,x^2)$,
such that
\begin{eqnarray*} 
2x &=& 2^i, \mbox{~~i.e.,} \\
x &=& 2^i/2,
\end{eqnarray*}
which implies that the point of tangency is the point $(2^{i-1},2^{2(i-1)})$.
Accordingly, this $(j+1)$-st call 
of Quickhull will next make a call on a set of points that is only 
one fewer than that for the $j$-th call, until all that is left is $T$
and the triangle,
$((0,0),(1,1),(2,4))$.
Thus, since each call includes $T$, the total running time 
of Quickhull on $S'$ is at least 
\[
\sum_{i=0}^{h-3} n-i,
\]
which is $\Omega(nh)$.
\end{proof}

Note that the $n$ points used in the proof of Theorem~\ref{thm:lower} can
each be represented exactly in floating point using $O(\log n)$ bits.

\subsection{Expected-time Performance for Uniform Distributions}
Overmars and van~Leeuwen~\cite{overmars}
show that the deterministic Quickhull algorithm has an expected 
running time of $O(n)$ for $n$ points chosen uniformly at random 
from a bounded convex region.
In this subsection, we prove a similar result for 
the randomized Ray-shooting Quickhull algorithm.

\begin{theorem}
\label{thm:linear}
If $n$ points are chosen independently at
random from a uniform distribution in a bounded
convex region, $R$, then the expected running time 
of the randomized Ray-shooting Quickhull algorithm is $O(n)$.
\end{theorem}
\begin{proof}
For the sake of simplicity, our proof 
does not try to optimize the constant factor in the $O(n)$ bound.
Each recursive call in the randomized Ray-shooting Quickhull algorithm
(after the first) is defined by a subset of points
from $R$ contained in a bounding triangle, $pzr$.
Since the points in $R$ are chosen uniformly at random,
the total expected running of randomized Ray-shooting Quickhull is bounded
by the total area of all of these triangles.
Consider one such triangle, and,
w.l.o.g., let $\overline{pr}$ denote the base of this triangle.
Let $v$ denote the point inside the triangle, $pzr$,
that would be chosen by the deterministic Quickhull algorithm, i.e.,
the point farthest from the line $\overline{pr}$.
We distinguish two cases, depending on the how far away $v$ is
from $\overline{pr}$.
\begin{itemize}
\item
Case 1: the distance from $\overline{pr}$ to $v$ is at least $1/6$ the
distance from $\overline{pr}$ to $z$.
(See Figure~\ref{fig:proof}.)
Let $u$ and $w$, respectively, be the midpoints of the edges $\overline{pv}$
and $\overline{vr}$.
Then the triangle $uvw$ has one fourth the area of the triangle $pvr$; hence,
with probability $1/4$ the randomized Ray-shooting Quickhull 
algorithm will choose a pivot, $q$,
inside the triangle $uvw$ for this call.
If this occurs, then the randomized Ray-shooting Quickhull algorithm
will at least eliminate 
all the points in the triangle $pqr$, 
which has area at least one half the area of the 
triangle $pvr$; hence, this choice for $q$ eliminates 
at least $(1/6)/2=1/12$ of the points for this
call in expectation; hence, the two recursive calls are performed
on at most $(11/12)n$ points in expectation for this case.
\end{itemize}

\begin{figure}[hbt]
\centering
\includegraphics[width=3in, trim = 0.9in 3.6in 7.5in 1.4in, clip]{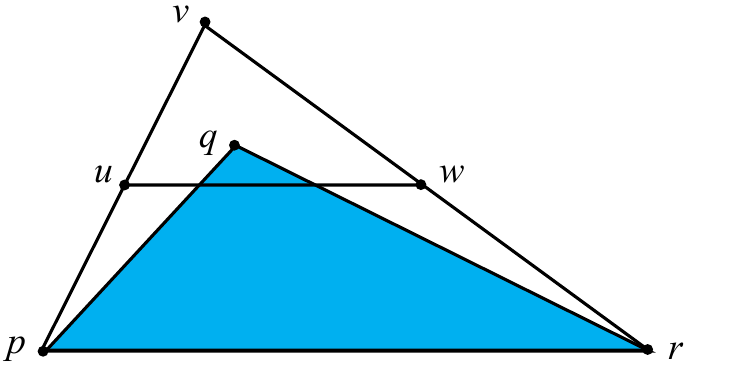}
\caption{\label{fig:proof} Illustrating case 1 of the proof of Theorem~\ref{thm:linear}.}
\end{figure}

\begin{itemize}
\item
Case 2:
the distance from $\overline{pr}$ to $v$ is less than $1/6$ the
distance from $\overline{pr}$ to $z$. To be conservative
for this case, let us consider $v$ being at distance exactly $1/6$
the distance from $\overline{pr}$ to $z$. 
(See Figure~\ref{fig:proof2}.)
We claim that the total area of the triangles for the 
two recursive calls for this case is at most a constant fraction, $\epsilon<1$,
of that for the triangle, $pzr$.
For the sake of considering a
worst case for creating large subtriangles, let $q'$ be the point
above the pivot, $q$, at distance $1/6$ the distance from $\overline{pr}$ to $z$.
Let $z_1$ be the point on $\overline{pz}$ such that $\overline{z_1q'}$
has the same slope as $\overline{zr}$, and
Let $z_2$ be the point on $\overline{zr}$ such that $\overline{q'z_2}$
has the same slope as $\overline{pz}$.
Thus, by convexity,
the slopes of the convex hull edges in $pzr$ must be between
the slope of $\overline{pz}$ and $\overline{zr}$; hence,
the triangles for the recursive calls for $pzr$ must exclude
the quadrilateral $(q',z_2,z,z_1)$.
Further, let $r'''$ be a point on the line $\overline{p'r'}$ at distance
$1/4d(p',r')$, and let $r''$ be the point directly below it on $\overline{pr}$.
Then the area of the quadrilateral $(p,p',r''',r'')$ and its twin (by symmetry)
on the right side are at most
$(11/12)\cdot(1/6)$ the area of the triangle $pzr$, 
whereas the area of the intersection
of $R$ with $pzr$ in this case must be at least $1/6$ the area of the 
triangle $pzr$. Thus, with probability at least $1/12$, the point $q'$ will be
between $r'''$ and its twin on the right.
At an extreme case, then, when $q'=r'''$, because the ratio of the area of 
two similar triangles is proportional to the square of 
the ratio of their corresponding sides, the sum of the areas of the triangles
$p'z_1q'$ and $q'z_2r'$ is at most $(1/4)^2+(3/4)^2=5/8$ that of the triangle
$p'zr'$. Thus,
the quadrilateral $(q',z_2,z,z_1)$
has area at least $3/8$ that of $p'zr'$, which is at least $(3/8)\cdot(5/6)^2$
that of $pzr$.
\end{itemize}
Thus, we can bound the expected running time, $T(n)$, of 
the randomized Ray-shooting
Quickhull algorithm for $R$ using the following recurrence,
for constants, $0<\delta,\epsilon<1$:
\[
T(n) \le (1-\delta)\cdot\left(T(n_1)+T(n_2)+cn\right)
         + \delta\left(T(n_1')+T(n_2')+cn\right),
\]
where $c\ge 1$ is a constant, $n_1+n_2=n$, and $n_1'+n_2'\le \epsilon n$.
Therefore, by an induction argument, $T(n)$ is $O(n)$.
\end{proof}

\begin{figure}[hbt]
\vspace*{-12pt}
\centering
\includegraphics[width=3in, trim = 0.9in 3.3in 7.5in 1.4in, clip]{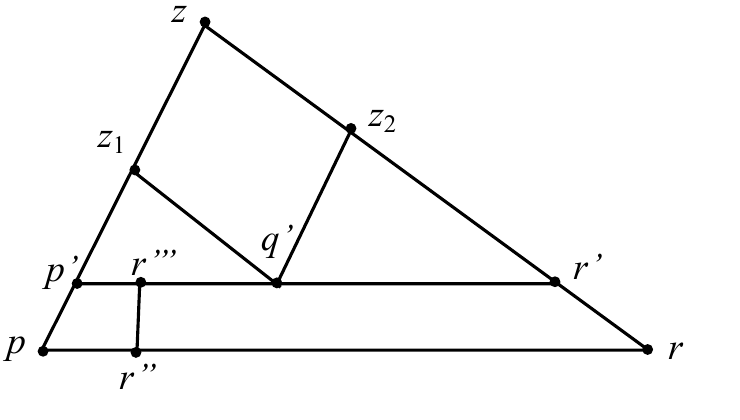}
\caption{\label{fig:proof2} Illustrating case 2 of the proof of Theorem~\ref{thm:linear}.}
\end{figure}

\subsection{Experiments}

In this section, we support our theoretical analysis with experimental results,
comparing the performance of our randomized Ray-shooting Quickhull algorithm
to the deterministic Quickhull algorithm.
The deterministic Quickhull algorithm is already known empirically to perform 
well in practice compared to other convex hull 
algorithms~\cite{gamby2018convex, mitura, overmars};
so we have restricted our experiments to be a head-to-head comparison
of the deterministic Quickhull and randomized Ray-shooting Quickhull
algorithms.
As we show below, our experiments provide empirical evidence that
the randomized version is
competitive with and in some cases outperforms the deterministic 
version.\footnote{Our implementation will be made available on GitHub once
   anonymity is no longer a concern.}

To maintain consistency with previous experimental work,
our implementation of the deterministic Quickhull algorithm was 
heavily adapted from existing 
sources~\cite{noauthor_quickhull_2023, joshi_anantjoshiczquickhull_2024}.
Also, because our experiments are focused on inputs that themselves 
have entropy, 
our implementation of the randomized 
Ray-shooting Quickhull algorithm skips the 
random permutation in line~\ref{permute}; see, e.g., 
Chung, Mitzenmacher, and Vadhan~\cite{chung2013simple} 
for additional support for this choice.

\subsubsection{Experimental Setup}
The algorithms were implemented in C++ and compiled with the same compiler and
optimization flags.
Our experiments were run on a machine with an Intel i5-1240P CPU and 8GB of RAM.
We tested our algorithm on five different input distributions.
A summary of the various input distributions can be found in 
Figure~\ref{tab:input_distributions}.

\begin{figure}[H]
  \centering
  \begin{tabular}{|c|c|}
    \hline
    {\bf Distribution} & {\bf Description} \\
    \hline
    Square & Points are uniformly distributed within a square \\
    \hline
    Circle & Points are uniformly distributed within a circle \\
    \hline
    On Circle & Points are uniformly distributed on a circle \\
    \hline
    Quad & Points are of the form of $(x, x^2)$ such that $x$ is uniformly distributed \\
    \hline
    Worst & For $i = 1, \dots, n$ points take the form of $(2^i, 2^{2i})$ randomly shuffled \\
    \hline
  \end{tabular}
  \caption{Descriptions of the various input distributions.}
  \label{tab:input_distributions}
\end{figure}

For each experiment, with the various input distributions and sizes, we ran
1000 trials per algorithm and took the average time of each in milliseconds.
Intuitively, we should expect the Square and Circle distributions to
favor the deterministic Quickhull algorithm, which runs in $O(n)$ 
expected time for these distributions~\cite{overmars} with low overhead.
The On-Circle and Quad distributions shouldn't favor either algorithm
asymptotically, as both should run in $O(n\log n)$ expected time 
for these input distributions.
The Worst distribution, on other hand, should favor 
the randomized Ray-shooting Quickhull algorithm, since it is
the distribution of Theorem~\ref{thm:lower}.

\subsubsection{Results}
The results for the experiment containing the points randomly chosen
within the unit circle and square are shown in Figure \ref{fig:shapes}.
In spite of the Square and Circle distributions being designed to
favor the deterministic Quickhull algorithm,
we see here that the randomized Ray-shooting Quickhull
algorithm performs comparably to the
deterministic algorithm, only being slower by a small 
constant factor.\footnote{All of our plots are log-log plots.}
For both these distributions, one would expect the deterministic algorithm
to perform slightly better than the randomized version, which is supported
by our results.

\begin{figure}[H]
  \centering
  \begin{subfigure}{0.3\textwidth}
    \includegraphics[width=\textwidth]{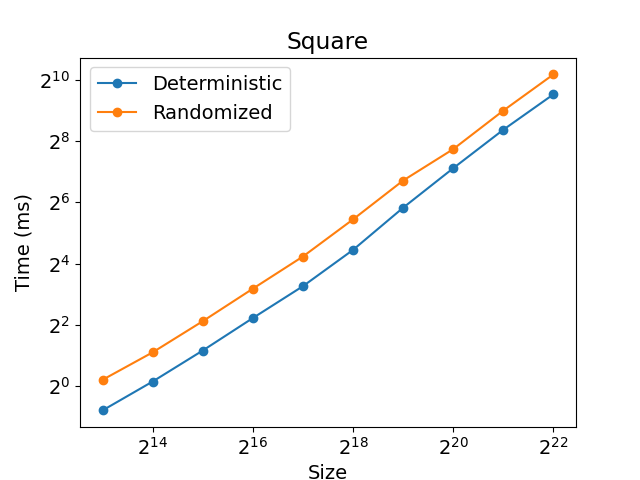}
  \end{subfigure}
  ~
  \begin{subfigure}{0.3\textwidth}
    \includegraphics[width=\textwidth]{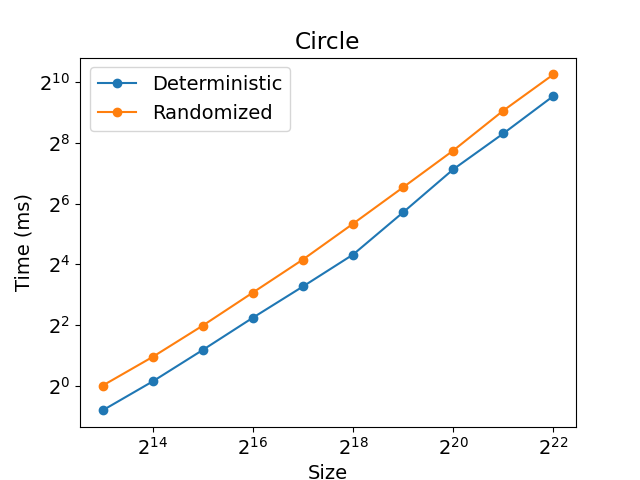}
  \end{subfigure}
  \caption{Distributions where points lie within the unit square and circle.}
  \label{fig:shapes}
\end{figure}

For the On-Circle and Quad distributions, 
we see that the randomized algorithm
performs significantly better across all input sizes, which we
found surprising.
Finally, for the Worst distribution, we see that unsurprisingly
that the randomized Ray-shooting Quickhull algorithm is significantly faster.
These running-time plots are shown in Figure~\ref{fig:worst}.

\begin{figure}[H]
  \centering
  \begin{subfigure}{0.3\textwidth}
    \includegraphics[width=\textwidth]{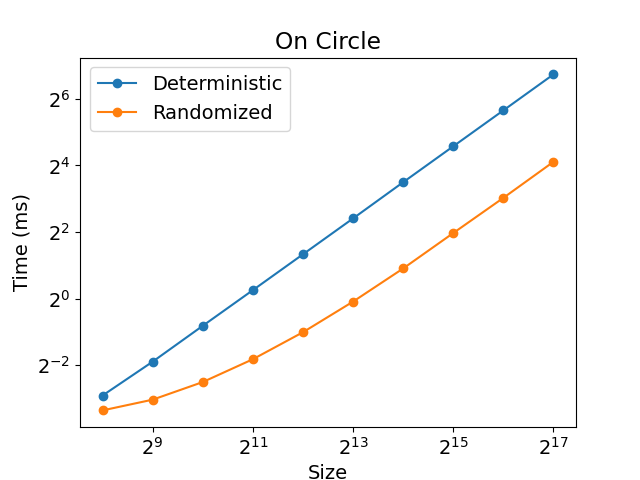}
  \end{subfigure}
  ~
  \begin{subfigure}{0.3\textwidth}
    \includegraphics[width=\textwidth]{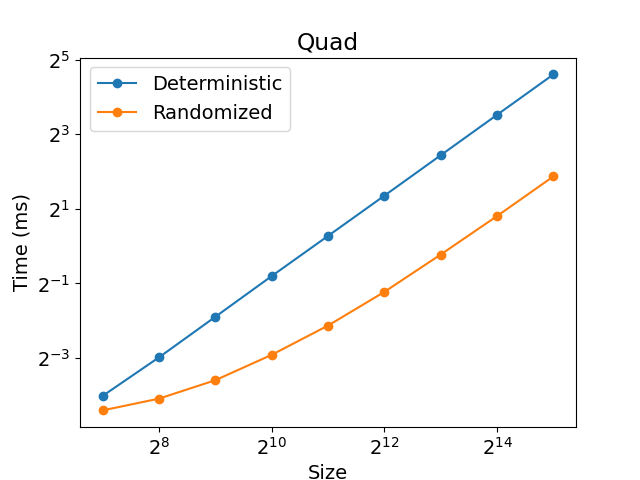}
  \end{subfigure}
  ~
  \begin{subfigure}{0.3\textwidth}
    \includegraphics[width=\textwidth]{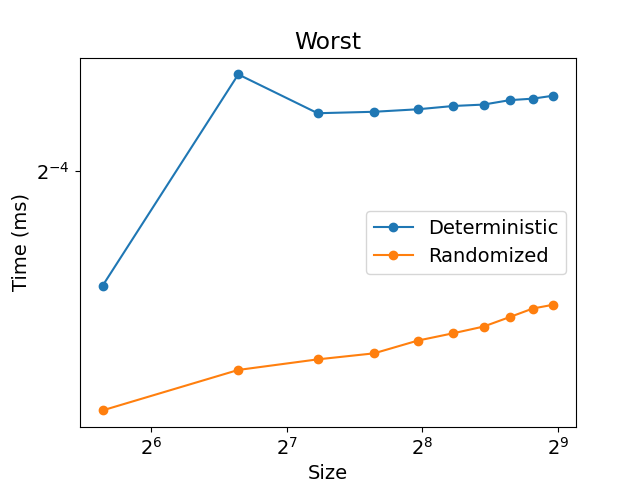}
  \end{subfigure}
  \caption{Performance for the On-Circle, Quad, and Worst distributions.}
  \label{fig:worst}
\end{figure}

\clearpage

\bibliographystyle{plainurl}
\bibliography{refs}

\end{document}